\documentclass[preprint,12pt]{article}
\usepackage{bbm}
\usepackage{amsfonts}
\usepackage{mathrsfs}
\usepackage{amssymb}
\usepackage{amsmath}
\usepackage{amsthm}
\usepackage{color}
\usepackage{url}
\usepackage{hyperref}
\usepackage{multirow}
\usepackage{bbding}
\usepackage[ruled,vlined]{algorithm2e}
\usepackage{latexsym}
\usepackage[final]{graphicx}
\usepackage{cases}
\usepackage{amsthm}

\usepackage{lastpage}
\usepackage{epsfig}
\usepackage[misc]{ifsym}
\usepackage{lineno}

\newtheorem{definition}{Definition}
\newtheorem{theorem}{Theorem}
\newtheorem{lemma}{Lemma}

\begin{document}

\begin{center}

\textbf{\large{Asymptotically Ideal Conjunctive Hierarchical Secret Sharing Scheme Based on CRT for Polynomial Ring}}

\vspace{3mm}
Jian Ding$^{1,2}$, Cheng Wang$^2$, Hongju Li$^{1,3}$, Cheng Shu$^2$, Haifeng Yu$^2$

\vspace{3mm}
\footnotesize{
$^1$School of Mathematics and Big Data, Chaohu University, Hefei 238024, China\\
$^2$School of Mathematics and Statistics, Hefei University, Hefei 230000, China\\
$^3$}School of Computer Science and Technology, University of Science and Technology of China, Hefei 230026, China\\

\footnotetext{\footnotesize {This research was in part supported by Projects of Chaohu University under Grant
KYQD-202220, Grant 2024cxtd147, Grant kj22zdjsxk01 and Grant hxkt20250327, University Natural Science Research Project of Anhui Province under Grant 2024AH051324 and Grant 2025AHGXZK30056.}}

\footnotetext{\footnotesize {\noindent Corresponding author: Jian Ding}}

\footnotetext{\footnotesize {E-mail address: dingjian\_happy@163.com (J. Ding)}}
\end{center}

\noindent\textbf{Abstract}:
Conjunctive Hierarchical Secret Sharing (CHSS) is a type of secret sharing that divides participants into multiple distinct hierarchical levels, with each level having a specific threshold. An authorized subset must simultaneously meet the threshold of all levels. Existing Chinese Remainder Theorem (CRT)-based CHSS schemes either have security vulnerabilities or have an information rate lower than $\frac{1}{2}$. In this work, we utilize the CRT for polynomial ring and one-way functions to construct an asymptotically perfect CHSS scheme. It has computational security, and permits flexible share sizes. Notably, when all shares are of equal size, our scheme is an asymptotically ideal CHSS scheme with an information rate one.

\noindent\textbf{Keywords}: secret sharing, conjunctive hierarchical secret sharing, asymptotically ideal secret sharing, Chinese Remainder Theorem, information rate

\section {Introduction}

Secret sharing (SS) schemes were independently introduced by Shamir \cite{Shamir1979} and Blakley \cite{Blakley1979} in 1979. A typical SS scheme consists of two phases: a share generation phase and a secret reconstruction phase. In the share generation phase, the dealer divides the secret into multiple shares and distributes them to the participants. During the secret reconstruction phase, any authorized set of participants can reconstruct the secret, while any unauthorized set cannot. The collection of all authorized subsets is known as the \emph{access structure}. An SS scheme is said to be \emph{perfect} if any given unauthorized subset learns no information about the secret. \emph{Information rate} is a key efficiency measure in SS schemes, defined as the ratio between the entropy of the secret and the maximum entropy among all shares. An SS scheme is termed \emph{ideal} if it is perfect and has an information rate one. It is referred to as \emph{asymptotically ideal} if it is asymptotically perfect and its information rate approaches 1 as the secret size increases.

The $(t,n)$-threshold scheme is a form of secret sharing in which any subset of participants of size at least $t$ is authorized to reconstruct the secret, whereas any subset with $(t-1)$ or fewer participants is unauthorized. In such schemes, all participants hold identical privileges during the secret reconstruction phase. However, many practical applications require assigning different privilege levels to participants. To meet this need, the \emph{Conjunctive Hierarchical Secret Sharing} (CHSS) scheme was introduced \cite{Tassa2007}. It organizes participants into multiple hierarchical levels, each with its own threshold. An authorized subset must simultaneously satisfy the thresholds of every level. Based on Birkhoff interpolation, Tassa \cite{Tassa2007} constructed an ideal and unconditionally secure CHSS scheme, though the dealer had to check the non-singularity of an exponential number of matrices. Gyarmati et al. \cite{2025Gyarmati} later reduced the size of the underlying finite field in Tassa's scheme using finite geometry.

Employing polymatroid theory, Chen et al. \cite{Chenqi2022} presented an ideal CHSS scheme with an explicit construction. It only requires the dealer to perform polynomial-time computations to determine certain parameters. Yuan et al. \cite{YuanJiaotong2022} proposed a multi-secret CHSS scheme based on homogeneous linear relations and one-way functions. The scheme achieved a smaller share size compared to Chen et al. \cite{Chenqi2022}, but demanded a large number of public values. More recently, Chintamani et al. \cite{Chintamani2024} designed another multi-secret CHSS scheme using elliptic curves and bilinear pairings, which supports an arbitrary number of participants and places no restriction on the number of secrets that can be shared. Because both Yuan et al. \cite{YuanJiaotong2022} and Chintamani et al. \cite{Chintamani2024} rely on one-way functions or elliptic curves, their schemes are asymptotically ideal when only one secret is shared.

Chinese Remainder Theorem (CRT)-based SS schemes offer a notable advantage: the flexibility to assign shares of varying sizes to different participants \cite{Asmuth-Bloom1983, Ning2018, Ding2023}. Consequently, this work focuses on CRT-based CHSS schemes. Ersoy et al. \cite{Oguzhan-Ersoy2016} constructed a CHSS scheme using the CRT for integer ring and employing one-way functions. However, their scheme suffers from a low information rate-specifically, less than $\frac{1}{2}$. Subsequently, Tiplea et al. \cite{Tiplea2021} proposed an asymptotically ideal CHSS scheme also based on the CRT for integer ring, whose information rate approaches 1. Nevertheless, this scheme has been proven insecure \cite{Hongjuarx}.

\emph{Our contributions}. By using the CRT for polynomial ring and one-way functions, we construct an asymptotically ideal CHSS scheme with an information rate one. It has flexible share sizes, but the ideal CHSS schemes of Tassa \cite{Tassa2007}, Gyarmati et al. \cite{2025Gyarmati}, Chen et al. \cite{Chenqi2022} and the asymptotically ideal CHSS schemes of Yuan et al. \cite{YuanJiaotong2022} and Chintamani et al. \cite{Chintamani2024} cannot. Compared with the CRT-based CHSS schemes of Tiplea et al. \cite{Tiplea2021} and Ersoy et al. \cite{Oguzhan-Ersoy2016}, our scheme is a secure and asymptotically ideal DHSS scheme (See Table 1).
     \begin{table}[!htb]
  \centering
   {\bf Table 1.}  Conjunctive hierarchical secret sharing schemes. \\
  \begin{tabular}{cccccc}
  \hline
    \multirow{2}*{Schemes}  & \multirow{2}*{Security}    & Flexible             &\multirow{2}*{Perfectness}  &Information\\
    ~                        &~                        & share sizes      &~                          &rate $\rho$\\
  \hline
  \hline
  \cite{Tassa2007}              &Unconditional                      &No                      &Yes                       &$\rho=1$\\
  \hline
  \cite{2025Gyarmati}           &Unconditional                       &No                      &Yes                       &$\rho=1$\\
  \hline
  \cite{Chenqi2022}             &Unconditional                       &No                      &Yes                       &$\rho=1$\\
  \hline
  \multirow{2}*{\cite{YuanJiaotong2022}}  &\multirow{2}*{Computational}   &\multirow{2}*{No}  & \multirow{2}*{Asymptotic}  &$\rho=1$ \\
  ~                             &~                                    &~                      &~                         &(one secret was shared)\\
  \hline
  \cite{Chintamani2024}        &\multirow{2}*{Computational}   &\multirow{2}*{No}  & \multirow{2}*{Asymptotic}  &$\rho=1$ \\
  ~                             &~                                    &~                      &~                         &(one secret was shared)\\
  \hline
   \cite{Tiplea2021}            &No                       &Yes                     &No                        &Approaches 1\\
  \hline
  \cite{Oguzhan-Ersoy2016}      &Computational                       &Yes                  &Asymptotic                   &$\rho<\frac{1}{2}$\\
  \hline
  Our scheme                    &Computational                       &Yes                   &Asymptotic                  &$\rho=1$\\
  \hline
  \end{tabular}
 \end{table}

\emph{Related work}. Recently, Yang et. \cite{Yangjing2024} constructed an ideal disjunctive hierarchical secret sharing scheme based on CRT for polynomial ring. It follows a similar approach to that of Tiplea et al. \cite{Tiplea2021}, and has been proven insecure \cite{Hongjuarx}.

\emph{Paper organization}. The remainder of this paper is structured as follows. Section \ref{Sec: Prelim} introduces the necessary preliminaries. In Section \ref{Sec: our scheme}, we present the construction of an asymptotically ideal CHSS scheme and provide its security analysis. Section \ref{Sec: conclusion} concludes our work.

\section{Preliminaries}\label{Sec: Prelim}
This section will introduce the Chinese Reminder Theorem (CRT) for polynomial ring  as well as related concepts of secret sharing. Let $n, N_1, N_2$ be positive integers such that $N_1<N_2$.  Denote by $[n]=\{1,2,\ldots,n\}$ and $[N_1, N_2]=\{N_1,N_1+1,\ldots,N_2\}$. Let $\mathsf{H}(\mathbf{X})$ be the Shannon entropy of the random variable $\mathbf{X}$, and let $\mathsf{H}(\mathbf{X}|\mathbf{Y})$ be the conditional entropy of $\mathbf{X}$ given $\mathbf{Y}$.

\subsection{CRT for polynomial ring}

\begin{lemma}[CRT for polynomial ring, \cite{Ning2018,Ding2023}]\label{le: CRT}
Let $\mathbb{F}$ be a finite field, and let $m_1(x),m_2(x),\\\ldots,m_n(x)\in \mathbb{F}[x]$ be pairwise coprime polynomials.
For any given polynomials $g_1(x),g_2(x),\\\ldots,g_n(x)\in \mathbb{F}[x]$, consider the system of congruences
                  \[ y(x)\equiv g_i(x) \pmod {m_i(x)},~ \mathrm{for~ all}~ i\in [n].\]
Then its solutions satisfies
          \[y(x)\equiv\sum\limits_{i=1}^{n}\lambda_i(x)M_i(x)g_i(x) \pmod {M(x)},\]
where
        \[M(x)=\prod\limits_{i=1}^{n}m_i(x), M_i(x)=M(x)/m_i(x),~\mathrm{and}~\lambda_i(x) M_i^{-1}(x)\equiv 1\pmod {m_i(x)}.\]
Moreover, if the degree of $y(x)$ satisfies $\deg(y(x))<\deg(M(x))$, the solution is unique and can be written as
          \[y(x)=\sum\limits_{i=1}^{n}\lambda_i(x)M_i(x)g_i(x) \pmod {M(x)}.\]
\end{lemma}

\subsection{Secret sharing}

\begin{definition}[Secret sharing scheme]
Let $\mathcal{P}=\{P_1, P_2, \ldots, P_n\}$ denote the set of all $n$ participants. A Secret Sharing (SS) scheme consists of two phases: a share generation phase and a secret reconstruction phase, described as follows.
\begin{itemize}
\item \textbf{Share Generation Phase.} The dealer distributes shares to the participants via the mapping
                      \[\mathsf{SHARE}\colon\mathcal{S}\times\mathcal{R}\mapsto\mathcal{C}_1\times\mathcal{C}_2\times\dotsb\times\mathcal{C}_n,\]
   where $\mathcal{S}$ is the secret space, $\mathcal{C}_i$ is the share space of participant $P_i$, and $\mathcal{R}$ is the set of random inputs.

\item \textbf{Secret Reconstruction Phase.} Any authorized subset $\mathcal{A} \subseteq \mathcal{P}$ can reconstruct the secret using their shares through the mapping
                     \[\mathsf{RECON}\colon\prod_{P_i\in I}\mathcal{C}_i\mapsto\mathcal{S},\]
  Conversely, any unauthorized subset should not be able to reconstruct the secret.
\end{itemize}
\end{definition}

In this work, we identify each number $i$ with the $i$-th participant $P_i$, which implies that $[n] = \mathcal{P}$. An SS scheme is said to be \emph{perfect} if any given unauthorized subset learns no informtation about the secret.

\begin{definition}[Information rate, \cite{Ning2018}]\label{def: Information rate}
Let $\mathbf{S}$ and $\mathbf{C}_i$ be random variables representing the secret and the share of the $i$-th participant, respectively. The information rate of an SS scheme is defined as
                                \[\rho=\frac{\mathsf{H}(\mathbf{S})}{\max_{i\in [n]}^{}{\mathsf{H}(\mathbf{C}_i)}}.\]
\end{definition}

For any perfect SS scheme, the information rate satisfies $\rho \leq 1$. The scheme is called ideal if it is perfect and attains an information rate one.

\begin{definition}[Conjunctive hierarchical secret sharing scheme, \cite{Tassa2007}]\label{def: Conjunctive secret sharing}
Let $\mathcal{P}$ be a set of $n$ participants, partitioned by the dealer into $m$ disjoint subsets $\mathcal{P}_1, \mathcal{P}_2, \ldots, \mathcal{P}_m$ such that
           \[\mathcal{P}=\cup_{\ell=1}^{m} \mathcal{P}_{\ell},~and~\mathcal{P}_{\ell_1}\cap\mathcal{P}_{\ell_2}=\varnothing~ for~any~ 1\leq \ell_1<\ell_2\leq m.\]
For a threshold sequence $t_1, t_2, \ldots, t_m$ such that $1\leq t_1< t_2<\cdots<t_m\leq n$, the Conjunctive Hierarchical Secret Sharing (CHSS) scheme with the given threshold sequence is an SS scheme such that the following correctness and privacy are satisfied.
\begin{itemize}
  \item \textbf{Correctness}. The secret can be reconstructed by any given element of
       \[\Gamma=\{\mathcal{A} \subseteq \mathcal{P}: |\mathcal{A}\cap(\bigcup_{w=1}^{\ell}\mathcal{P}_{w})|\geq t_{\ell} ~\mathrm{for~all}~\ell \in [m]\},\]
       where $|\cdot|$ is the cardinality of a set.
  \item \textbf{Privacy}. The secret cannot be reconstructed by any given $\mathcal{B}\notin \Gamma$.
\end{itemize}
A CHSS scheme is said to be ideal if it is perfect and the information rate is one.
\end{definition}

\begin{definition}[Asymptotically ideal CHSS scheme, \cite{Quisquater2002}]\label{def: Asymptotically Ideal CHSS}
A CHSS scheme with secret space $\mathcal{S}$ and share spaces $\mathcal{C}_i$ for $i \in [n]$ is said to be asymptotically ideal if it satisfies the following properties:
\begin{itemize}
\item \textbf{Asymptotic perfectness.} For every $\epsilon_1 > 0$, there exists a positive integer $\sigma_1$ such that for all $\mathcal{B} \notin \Gamma$ and $|\mathcal{S}| > \sigma_1$,
                       \[\Delta(|\mathcal{S}|)=\mathsf{H}(\mathbf{S})-\mathsf{H}(\mathbf{S}|\mathbf{V}_{\mathcal{B}})\leq \epsilon_1,\]
   where $\mathsf{H}(\mathbf{S}) \neq 0$, and $\mathbf{S}$, $\mathbf{V}_{\mathcal{B}}$ are random variables representing the secret and the collection of the shares of $\mathcal{B}$, respectively.

\item \textbf{Asymptotic maximum information rate.} For every $\epsilon_2 > 0$, there exists a positive integer $\sigma_2$ such that for all $|\mathcal{S}| > \sigma_2$,
      \[\frac{\max_{i\in [n]}^{}{\mathsf{H}(\mathbf{S}_i)}}{\mathsf{H}(\mathbf{S})}\leq 1+\epsilon_2.\]
\end{itemize}
\end{definition}


\section{An asymptotically ideal CHSS scheme}\label{Sec: our scheme}
In this section, we present an asymptotically ideal CHSS scheme constructed by applying the Chinese Remainder Theorem (CRT) for polynomial ring along with one-way functions. The detailed scheme is described in Subsection \ref{Subsec:our DHSS}, and its security analysis is provided in Subsection \ref{Subsec: security of our DHSS}.

\subsection{Our scheme}\label{Subsec:our DHSS}

Let $\mathcal{P}$ be a set of $n$ participants partitioned into $m$ disjoint subsets $\mathcal{P}_1, \mathcal{P}_2, \ldots, \mathcal{P}_m$. For each $\ell \in [m]$, denote by $n_{\ell}= |\mathcal{P}_{\ell}|$ and $N_{\ell} = \sum_{w=1}^{\ell} n_w$. Let $t_1, t_2, \ldots, t_m$ be a strictly increasing threshold sequence satisfying
              \[1\leq t_1<t_2<\cdots<t_m,~\mathrm{and}~t_{\ell}\leq N_{\ell}~\mathrm{for}~\ell\in [m].\]
Let $\lfloor x \rfloor$ be the largest integer no greater than $x$. Our scheme comprises two phases: a share generation phase and a secret reconstruction phase.

\textbf{1) Share Generation Phase}. Let $p$ be a prime integer, and $\mathbb{F}_p$ be a finite field with $p$ elements. Denote by $\mathcal{S}=\{g(x)\in \mathbb{F}_p[x]:\deg(g(x))<d_0\}$ the secret space, where $d_0$ is a publicly known positive integer.
\begin{itemize}
      \item \emph{Identities}. The dealer sets $m_0(x) = x^{d_0}$, and selects a set of publicly known polynomials $m_i(x) \in \mathbb{F}_p[x], i \in [n]$ that satisfy the following three conditions:
          \begin{itemize}
              \item[(i)] $m_0(x), m_1(x),\ldots, m_n(x)$ are pairwise coprime polynomials over $\mathbb{F}_p$.
             \item[(ii)] Let $d_i = \deg(m_i(x)), i \in [n]$, then $d_0 \leq d_1 \leq d_2 \leq \cdots \leq d_n$.
             \item[(iii)] $d_0+\sum\limits_{i=n-t_{\ell}+2}^{n}d_i\leq \sum\limits_{i=1}^{t_{\ell}}d_i$ for all $\ell\in [m]$.
           \end{itemize}
   \item \emph{Shares}. For any $s(x)\in \mathcal{S}$, the dealer chooses random polynomials $s_{1}(x),s_{2}(x),\ldots,\\s_{m}(x) \in \mathcal{S}$ such that $s(x)=\sum_{\ell=1}^{m}s_{\ell}(x)\in \mathbb{F}_p[x]$. After that, the dealer selects random polynomials
             \[\alpha_{\ell}(x)\in \mathcal{G}_{\ell} =\{g(x)\in \mathbb{F}_{p}[x]:\deg(g(x))<(\sum\limits_{i=1}^{t_{\ell}}d_i)-d_0\}, \ell\in [m],\]
        and constructs polynomials $f_{\ell}(x)=s_{\ell}(x)+\alpha_{\ell}(x)x^{d_0}, \ell\in [m]$. Finally, the dealer selects random polynomials
                 \[r_{i}(x)= r_{i,0}+r_{i,1}x+\cdots+r_{i,d_i-1}x^{d_i-1}\in \mathbb{F}_{p}[x], i\in [N_{m-1}],\]
         and sends the share $c_i(x)$ to the $i$-th participant, where
         \[c_i(x)=\begin{cases}
            r_i(x),\mathrm{if}~i\in [N_{m-1}],\\
            f_m(x) \pmod {m_i(x)},\mathrm{if}~i\in [N_{m-1}+1,N_m].\\
         \end{cases}\]

        \item \emph{Hierarchy}. The dealer chooses $m$ distinct one-way functions $h_1(\cdot), h_2(\cdot), \dots, h_m(\cdot)$ and makes them public. Each function accepts an input of arbitrary length and produces an output of fixed length $\lfloor \log_{2} p \rfloor$. Then the dealer publishes the values
                  \[u_i^{(\ell)}(x)=(f_{\ell}(x)-H_{\ell}(c_i(x)))\pmod{m_i(x)}, \ell\in [m-1], i\in [N_{\ell}],\]
            and $u_i^{(m)}(x)=(f_{m}(x)-H_{m}(c_i(x)))\pmod{m_i(x)}$, $i\in [N_{m-1}]$, where
            \[H_{\ell}(c_i(x))=H_{\ell}(r_i(x))=h_{\ell}(r_{i,0})+h_{\ell}(r_{i,1})x+\cdots+h_{\ell}(r_{i,d_i-1})x^{d_i-1}\in \mathbb{F}_p[x], \ell \in [m].\]
\end{itemize}

\textbf{2) Secret Reconstruction Phase}. For any
      \[\mathcal{A}\in\Gamma=\{\mathcal{A} \subseteq \mathcal{P}: |\mathcal{A}\cap(\bigcup_{w=1}^{\ell}\mathcal{P}_{w})|\geq t_{\ell} ~\mathrm{for~all}~\ell \in [m]\},\]
and $\ell\in [m]$, participants of $\mathcal{A}^{(\ell)}=\mathcal{A}\cap (\cup_{w=1}^{\ell}\mathcal{P}_{w})$ determine
       \[f_{\ell}(x)=\sum\limits_{i\in \mathcal{A}^{(\ell)}}\lambda_{i,\mathcal{A}^{(\ell)}}(x)M_{i,\mathcal{A}^{(\ell)}}(x)c_i^{(\ell)}(x) \pmod {M_{\mathcal{A}^{(\ell)}}(x)},\]
where $M_{\mathcal{A}^{(\ell)}}(x)=\prod\limits_{i\in\mathcal{A}^{(\ell)}}m_i(x)$, $M_{i,\mathcal{A}^{(\ell)}}(x)=M_{\mathcal{A}^{(\ell)}}(x)/m_i(x)$, $\lambda_{i,\mathcal{A}^{(\ell)}}(x)\equiv M_{i,\mathcal{A}^{(\ell)}}^{-1}(x) \pmod {m_i(x)}$, and
        \[c_i^{(\ell)}(x)=\begin{cases}
            H_{\ell}(c_i(x))+u_i^{(\ell)}(x),~\mathrm{if}~\ell\in [m-1], i\in \mathcal{A}^{(\ell)}\subseteq [N_{\ell}],\\
            H_{m}(c_i(x))+u_i^{(m)}(x),~\mathrm{if}~\ell=m, i\in \mathcal{A}^{(\ell)}\cap [N_{m-1}],\\
            c_i(x),~\mathrm{if}~\ell=m, i\in \mathcal{A}^{(\ell)}\cap [N_{m-1}+1,N_m].\\
         \end{cases}\]
Consequently, the secret can be reconstructed by using
        \[s(x)=\sum_{\ell=1}^{m}s_{\ell}(x)\pmod{p}=\sum_{\ell=1}^{m}f_{\ell}(x)\pmod{m_0(x)}.\]

\subsection{Security analysis of our scheme}\label{Subsec: security of our DHSS}

This subsection is devoted to proving the correctness, asymptotic perfectness, and asymptotic maximum information rate of our scheme.

\begin{theorem}[Correctness]\label{Theorem:Correctness of our scheme} The secret can be reconstructed by any
       \[\mathcal{A}\in\Gamma=\{\mathcal{A} \subseteq \mathcal{P}: |\mathcal{A}\cap(\bigcup_{w=1}^{\ell}\mathcal{P}_{w})|\geq t_{\ell} ~\mathrm{for~all}~\ell \in [m]\}.\]
\end{theorem}

\begin{proof}
For any $\mathcal{A}\in\Gamma$ and $\ell\in [m]$, it holds that $|\mathcal{A}^{(\ell)}|=|\mathcal{A}\cap (\cup_{w=1}^{\ell}\mathcal{P}_{w})|\geq t_{\ell}$, and participants of $\mathcal{A}^{(\ell)}$ are able to compute
         \[c_i^{(\ell)}(x)=\begin{cases}
            H_{\ell}(c_i(x))+u_i^{(\ell)}(x),~\mathrm{if}~\ell\in [m-1], i\in \mathcal{A}^{(\ell)}\subseteq [N_{\ell}],\\
            H_{m}(c_i(x))+u_i^{(m)}(x),~\mathrm{if}~\ell=m, i\in \mathcal{A}^{(\ell)}\cap [N_{m-1}],\\
            c_i(x),~\mathrm{if}~\ell=m, i\in \mathcal{A}^{(\ell)}\cap [N_{m-1}+1,N_m],\\
         \end{cases}\]
where $\{c_i(x): i\in \mathcal{A}^{(\ell)}\}$ is the set of shares of $\mathcal{A}^{(\ell)}$, $\{u_i^{(\ell)}(x): \ell\in [m-1], i\in \mathcal{A}^{(\ell)}\subseteq [N_{\ell}]~\mathrm{or}~\ell=m, i\in \mathcal{A}^{(\ell)}\cap [N_{m-1}]\}$ is a set of publicly known polynomials, and $h_1(\cdot), h_2(\cdot), \dots,\\h_m(\cdot)$ are publicly known one-way functions. Therefore, participants of $\mathcal{A}^{(\ell)}$ establish a system of congruences
                 \begin{equation}\label{Eq:1}
                 f_{\ell}(x)\equiv c_{i}^{(\ell)}(x) \pmod{m_{i}(x)}, i\in \mathcal{A}^{(\ell)}.
                 \end{equation}

Recall that $d_0\leq d_1\leq d_2\leq\cdots\leq d_n$, $s(x)\in \mathcal{S}, \alpha_{\ell}(x)\in \mathcal{G}_{\ell}$, $f_{\ell}(x)=s_{\ell}(x)+\alpha_{\ell}(x)x^{d_0}$ and $|\mathcal{A}^{(\ell)}|\geq t_{\ell}$ for $\ell\in [m]$, the degree of the polynomial $f_{\ell}(x)$ satisfies
                \[\deg{(f_{\ell}(x))}<\sum\limits_{i=1}^{t_{\ell}}d_i\leq\sum\limits_{i\in \mathcal{A}^{(\ell)}}d_i.\]
According to the CRT for polynomial ring in Lemma \ref{le: CRT}, the system of congruences (\ref{Eq:1}) has a unique solution
  \begin{equation}\label{Eq:2}
   f_{\ell}(x)=\sum\limits_{i\in \mathcal{A}^{(\ell)}}\lambda_{i,\mathcal{A}^{(\ell)}}(x)M_{i,\mathcal{A}^{(\ell)}}(x)c_i^{(\ell)}(x) \pmod {M_{\mathcal{A}^{(\ell)}}(x)},
  \end{equation}
where $M_{\mathcal{A}^{(\ell)}}(x)=\prod\limits_{i\in\mathcal{A}^{(\ell)}}m_i(x)$, $M_{i,\mathcal{A}^{(\ell)}}(x)=M_{\mathcal{A}^{(\ell)}}(x)/m_i(x)$, and $\lambda_{i,\mathcal{A}^{(\ell)}}(x)\equiv M_{i,\mathcal{A}^{(\ell)}}^{-1}(x)\pmod {m_i(x)}$. Consequently, the secret can be reconstructed by using
        \[s(x)=\sum_{\ell=1}^{m}s_{\ell}(x)\pmod{p}=\sum_{\ell=1}^{m}f_{\ell}(x)\pmod{m_0(x)}.\]
\end{proof}

We now proceed to prove the \emph{asymptotic perfectness} of our scheme. Note that for each $\ell \in [m]$, the polynomial $f_{\ell}(x)$ is defined as $f_{\ell}(x)=s_{\ell}(x)+\alpha_{\ell}(x)x^{d_0}$, where $s_{\ell}(x) \in \mathcal{S} = \{g(x) \in \mathbb{F}_p[x] : \deg(g(x)) < d_0\}$ and $\alpha_{\ell}(x) \in \mathcal{G}_{\ell} = \{ g(x) \in \mathbb{F}_p[x] : \deg(g(x)) < (\sum_{i=1}^{t_{\ell}} d_i) - d_0 \}$. It follows that when $s_{\ell}(x)$ and $\alpha_{\ell}(x)$ are chosen uniformly at random from $\mathcal{S}$ and $\mathcal{G}_{\ell}$ respectively, the polynomial $f_{\ell}(x)$ is uniformly distributed over the set $\{f_{\ell}(x)\in \mathbb{F}_p[x]:\deg(f_{\ell}(x))<\sum\limits_{i=1}^{t_{\ell}}d_i\}$.

Let $\mathcal{B}$ denote the set of adversaries satisfying
        \[\mathcal{B}\subset \mathcal{P},~\mathrm{and}~\mathcal{B} \notin \Gamma=\{\mathcal{A} \subseteq \mathcal{P}: |\mathcal{A}\cap(\bigcup_{w=1}^{\ell}\mathcal{P}_{w})|\geq t_{\ell}~\mathrm{for~all}~ \ell \in [m]\}.\]
Define $\mathcal{B}^{(\ell)} = \mathcal{B} \cap (\bigcup_{w=1}^{\ell} \mathcal{P}_{w}), \ell\in [m]$. We consider the worst-case scenario, namely,
        \[|\mathcal{B}^{(m)}|< t_{m},~\mathrm{and}~|\mathcal{B}^{(\ell)}|\geq t_{\ell}~\mathrm{for~all}~\ell\in [m-1].\]
The adversaries in $\mathcal{B}$ know their shares, the upper bounds of the degrees of $f_{\ell}(x)\in \mathbb{F}_p[x], \ell\in [m]$, and all publicly known polynomials and one-way hash functions. Namely, the knowledge of $\mathcal{B}$ consists of the following five conditions (i)-(v).
\begin{itemize}
  \item[(i)] For every $\ell \in [m]$, $\deg(f_{\ell}(x)) < \sum_{i=1}^{t_{\ell}} d_i$.
  \item[(ii)] $s(x) \equiv \sum_{\ell=1}^{m} s_{\ell}(x) \pmod{p}$, and $s_{\ell}(x) \equiv f_{\ell}(x) \pmod{m_0(x)}$ for all $\ell \in [m]$.
  \item[(iii)] It holds that
                  \[f_{\ell}(x)\equiv(H_{\ell}(c_i(x))+u_i^{(\ell)}(x))\pmod{m_i(x)}, \ell\in [m-1], i\in \mathcal{B}^{(\ell)},\]
            and $f_{m}(x)\equiv(H_{m}(c_i(x))+u_i^{(m)}(x))\pmod{m_i(x)}$, $i\in \mathcal{B}^{(m-1)}$, i.e.,
             \begin{equation}\label{Eq:3}
         \left\{\begin{aligned}
            f_{1}(x)\equiv &(H_{1}(c_i(x))+u_i^{(1)}(x))\pmod{m_i(x)}~\mathrm{for~all}~i\in \mathcal{B}^{(1)},\\
            f_{2}(x)\equiv &(H_{2}(c_i(x))+u_i^{(2)}(x))\pmod{m_i(x)}~\mathrm{for~all}~i\in \mathcal{B}^{(2)},\\
                     &\vdots\\
            f_{m-1}(x)\equiv &(H_{m-1}(c_i(x))+u_i^{(m-1)}(x))\pmod{m_i(x)}~\mathrm{for~all}~i\in \mathcal{B}^{(m-1)},\\
            f_{m}(x)\equiv &(H_{m}(c_i(x))+u_i^{(m)}(x))\pmod{m_i(x)}~\mathrm{for~all}~i\in \mathcal{B}^{(m-1)}.\\
           \end{aligned} \right.
        \end{equation}
  \item[(iv)] For every $i\in \mathcal{B}\cap [N_{m-1}+1,N_m]$, it holds that $f_{m}(x)\equiv c_i(x)\pmod{m_i(x)}$.
  \item[(v)] For every $i\in [N_{m-1}]$ and $i\notin \mathcal{B}$, there is a level $\mathcal{P}_{\ell_1}$ such that $\ell_1\in [m-1]$ and $i\in \mathcal{P}_{\ell_1}$. There are polynomials $\widetilde{c}_i(x)\in \mathbb{F}_p[x]$ such that $\deg{(\widetilde{c}_i(x))}<d_i$, and $H_{\ell}(\widetilde{c}_i(x))=(f_{\ell}(x)-u_i^{(\ell)}(x))\pmod{m_i(x)}$ for all $\ell\in [\ell_1,m]$.
  \end{itemize}
Note that conditions (iii) and (v) ensure that the polynomials $f_1(x), f_2(x), \ldots, f_m(x) \in \mathbb{F}_p[x]$ satisfy publicly known polynomials $u_i^{(\ell)}(x)$.

Recall that $|\mathcal{B}^{(\ell)}|\geq t_{\ell}$ for all $\ell\in [m-1]$. From the expression (\ref{Eq:2}) and the system of congruences (\ref{Eq:3}), participants of $\mathcal{B}$ determine the polynomials $f_{\ell}(x)\in \mathbb{F}_p[x], \ell\in [m-1]$, and get $s_{\ell}(x)=f_{\ell}(x) \pmod{m_0(x)}, \ell \in [m-1]$. Therefore, the dealer will guess the secret by first selecting $\widetilde{f}_m(x)\in (\mathbb{F}_p[x]$ satisfying the following four conditions, and then computing $(\widetilde{f}_m(x)+\sum_{\ell=1}^{m-1} s_{\ell}(x))\pmod{m_0(x)}$.
 \begin{itemize}
   \item[(I)] Polynomials $f_{\ell}(x), s_{\ell}(x), \ell\in [m-1]$ are known by the adversaries in $\mathcal{B}$. Moreover, it holds that $\deg(\widetilde{f}_{m}(x)) < \sum_{i=1}^{t_{m}} d_i$
  \item[(II)] $\widetilde{f}_{m}(x)\equiv (H_{m}(c_i(x))+u_i^{(m)}(x))\pmod{m_i(x)}$ for all $i\in \mathcal{B}^{(m-1)}$.
  \item[(III)] For every $i\in \mathcal{B}\cap [N_{m-1}+1,N_m]$, it holds that $\widetilde{f}_{m}(x)\equiv c_i(x)\pmod{m_i(x)}$.
  \item[(IV)] For every $i\in [N_{m-1}]$ and $i\notin \mathcal{B}$, there is a level $\mathcal{P}_{\ell_1}$ such that $\ell_1\in [m-1]$ and $i\in \mathcal{P}_{\ell_1}$. There are polynomials $\widetilde{c}_i(x)\in \mathbb{F}_p[x]$ such that $\deg{(\widetilde{c}_i(x))}<d_i$, $H_{\ell}(\widetilde{c}_i(x))=(f_{\ell}(x)-u_i^{(\ell)}(x))\pmod{m_i(x)}$ for all $\ell\in [\ell_1,m-1]$, and $H_{m}(\widetilde{c}_i(x))=(\widetilde{f}_{m}(x)-u_i^{(m)}(x))\pmod{m_i(x)}$.
  \end{itemize}

\begin{lemma}\label{le: loss entropy}
Let $\mathcal{V}_{\mathcal{B}}$ and $\mathcal{V}^\prime_{\mathcal{B}}$ be the conditions (I) to (IV) and the conditions (I) to (III), respectively. For all $\epsilon_1>0$, there is a positive integer $\sigma_1$ such that for all $|\mathcal{S}|=p^{d_0}>\sigma_1$, it has that
                               \[0<\mathsf{H}(\mathbf{S}|\mathbf{V}^\prime_{\mathcal{B}})-\mathsf{H}(\mathbf{S}|\mathbf{V}_{\mathcal{B}})<\epsilon_1,\]
where $\mathbf{V}_{\mathcal{B}}$ and $\mathbf{V}^\prime_{\mathcal{B}}$ are random variables representing $\mathcal{V}_{\mathcal{B}}$ and $\mathcal{V}^\prime_{\mathcal{B}}$, respectively.
\end{lemma}
\begin{proof}
Since the shares $c_i(x), i\in [N_{m-1}]$ are randomly selected by the dealer and functions $h_{\ell}(\cdot), \ell\in [m]$ are distinct one-way functions, then the condition (IV) can eliminate a polynomial $\widetilde{f}_m(x)\in \mathbb{F}_p[x]$ with a negligible probability when $|\mathcal{S}|$ is big enough. This shows that $\mathsf{H}(\mathbf{S}|\mathbf{V}^\prime_{\mathcal{B}})-\mathsf{H}(\mathbf{S}|\mathbf{V}_{\mathcal{B}})$ is negligible when $|\mathcal{S}|$ is big enough. This gives the proof.
\end{proof}

Next, we will compute the conditional entropy $\mathsf{H}(\mathbf{S}|\mathbf{V}_{\mathcal{B}})$. Denote by
\begin{equation}\label{Eq: definiton F of CHSS}
            \mathcal{F}=\{\widetilde{f}_m(x)\in \mathbb{F}_p[x]: ~\mathrm{the~conditions~(I)~to~(III)~}\mathrm{are~satisfied}.\}
       \end{equation}
For any given $s(x)\in \mathcal{S}$, we will compute how many possible $\widetilde{f}_m(x)\in \mathcal{F}$ such that $s(x)=(\widetilde{f}_m(x)+\sum_{\ell=1}^{m-1} s_{\ell}(x))\pmod{m_0(x)}$ in Lemma \ref{le: preimage}. After that, we will determine the cardinality of $\mathcal{F}$ in Lemma \ref{le: cardinality of F}. Based on Lemmas \ref{le: loss entropy}, \ref{le: preimage} and \ref{le: cardinality of F}, we determine the conditional entropy $\mathsf{H}(\mathbf{S}|\mathbf{V}^\prime_{\mathcal{B}})$ and give the proof of the asymptotic perfectness of our scheme in Theorem \ref{Th: our perfectness}.

\begin{lemma}\label{le: preimage}
 Let $\Phi$ be the mapping
 \[\Phi: \mathcal{F}\mapsto \mathcal{S}, \widetilde{f}_m(x)\mapsto (\widetilde{f}_m(x)+\sum_{\ell=1}^{m-1} s_{\ell}(x))\pmod{m_0(x)}.\]\
 For any $s(x)\in \mathcal{S}$, let
        \begin{displaymath}
        \begin{aligned}
            \Phi^{-1}(s(x))=&\{\widetilde{f}_m(x) \in \mathcal{F}: s(x)=(\widetilde{f}_m(x)+\sum_{\ell=1}^{m-1} s_{\ell}(x))\pmod{m_0(x)}\}.
        \end{aligned}
        \end{displaymath}
The cardinality of the set $\Phi^{-1}(s(x))$ is $| \Phi^{-1}(s(x))|=p^{\theta}$, where
            \[\theta=\sum\limits_{i=1}^{t_m}d_i-\sum\limits_{i\in \mathcal{B}}d_i-d_0 \geq 0.\]
\end{lemma}
\begin{proof}
For any $s(x)\in \mathcal{S}$, and $\widetilde{f}(x)\in \Phi^{-1}(s(x))$, it holds that
  \begin{displaymath}
         \left\{\begin{aligned}
            &\widetilde{f}_m(x) \equiv (s(x)-\sum_{\ell=1}^{m-1} s_{\ell}(x)) \pmod{m_0(x)},\\
            &\widetilde{f}_m(x) \equiv  c_i^{(m)}(x) \pmod{m_i(x)}~\mathrm{for~all}~i\in \mathcal{B},\\
           \end{aligned} \right.
  \end{displaymath}
where
\[c_i^{(m)}(x)=\begin{cases}
            H_{m}(c_i(x))+u_i^{(m)}(x),~\mathrm{if}~i\in \mathcal{B}\cap [N_{m-1}],\\
            c_i(x),~\mathrm{if}~ i\in \mathcal{B}\cap [N_{m-1}+1,N_m].\\
         \end{cases}\]
Based on Lemma \ref{le: CRT}, we have
\[\widetilde{f}_m(x)\equiv\sum\limits_{i\in \mathcal{\widetilde{B}}}\lambda_{i,\mathcal{\widetilde{B}}}(x)M_{i,\mathcal{\widetilde{B}}}(x)c_i^{(m)}(x)\pmod {M_{\mathcal{\widetilde{B}}}(x)},\]
where $\mathcal{\widetilde{B}}=\{0\}\cup \mathcal{B}$, $c_0^{(m)}(x)=s(x)-\sum_{\ell=1}^{m-1} s_{\ell}(x)$,
$M_{\mathcal{\widetilde{B}}}(x)=\prod\limits_{i\in\mathcal{\widetilde{B}}}m_i(x)$, $M_{i,\mathcal{\widetilde{B}}}(x)=M_{\mathcal{\widetilde{B}}}(x)/m_i(x)$, and $\lambda_{i,\mathcal{\widetilde{B}}}(x)\equiv M_{i,\mathcal{\widetilde{B}}}^{-1}(x)\pmod {m_i(x)}$.

Denote by $\widetilde{f}_{m,\mathcal{\widetilde{B}}}(x)=\sum\limits_{i\in \mathcal{\widetilde{B}}}\lambda_{i,\mathcal{\widetilde{B}}}(x)M_{i,\mathcal{\widetilde{B}}}(x)c_i^{(\ell)}(x)\pmod {M_{\mathcal{\widetilde{B}}}(x)}$, then
    \[\deg(\widetilde{f}_{m,\mathcal{\widetilde{B}}}(x))<\deg(M_{\mathcal{\widetilde{B}}}(x))=\sum\limits_{i\in \mathcal{\widetilde{B}}}d_i,\]
and
      \begin{equation}\label{Eq: preimage F}
     \widetilde{f}_{m}(x)\equiv \widetilde{f}_{m,\mathcal{\widetilde{B}}}(x)\pmod{M_{\mathcal{\widetilde{B}}}(x)}.
     \end{equation}
Consequently, there exists $k_{\mathcal{\widetilde{B}}}(x)\in \mathbb{F}_p[x]$ such that
     \[\widetilde{f}_{m}(x)=\widetilde{f}_{m,\mathcal{\widetilde{B}}}(x)+k_{\mathcal{\widetilde{B}}}(x) M_{\mathcal{\widetilde{B}}}(x).\]

On the one hand, $\widetilde{f}_{m}(x) \in \mathcal{F}$ and $\deg(\widetilde{f}_{m}(x))<\sum\limits_{i=1}^{t_{m}}d_i$ from the expression (\ref{Eq: definiton F of CHSS}). Since $\deg(\widetilde{f}_{m,\mathcal{\widetilde{B}}}(x))<\deg(M_{\mathcal{\widetilde{B}}}(x))$, then we obtain that
        \[\deg(k_{\mathcal{\widetilde{B}}}(x))<\sum\limits_{i=1}^{t_m}d_i-\sum\limits_{i\in \mathcal{\widetilde{B}}}d_i.\]
Therefore, there are $p^{\theta}$ different choices for $k_{\mathcal{\widetilde{B}}}(x)$, where
     \[\theta=\sum\limits_{i=1}^{t_m}d_i-\sum\limits_{i\in \mathcal{\widetilde{B}}}d_i=\sum\limits_{i=1}^{t_m}d_i-\sum\limits_{i\in \mathcal{B}}d_i-d_0 \geq 0.\]
On the other hand, each choice of polynomials
    \[k_{\mathcal{\widetilde{B}}}(x)\in \mathbb{F}_p[x], \deg(k_{\ell,\mathcal{\widetilde{B}}^{(\ell)}}(x))<\theta\]
produces a different $\widetilde{f}_{m}(x)\in \mathbb{F}_p[x]$ satisfying $\deg(f_{m}(x))<\sum\limits_{i=1}^{t_{m}}d_i$ and the expression (\ref{Eq: preimage F}), namely, $\widetilde{f}_{m}(x)\in \mathcal{F}$. Consequently, the cardinality of the set $\Phi^{-1}(s(x))$ is $| \Phi^{-1}(s(x))|=p^{\theta}$, where
            \[\theta=\sum\limits_{i=1}^{t_m}d_i-\sum\limits_{i\in \mathcal{B}}d_i-d_0 \geq 0.\]
\end{proof}

\begin{lemma}\label{le: cardinality of F}
 The cardinality of $\mathcal{F}$ is $|\mathcal{F}|=p^{\theta+d_0}$.
\end{lemma}

\begin{proof}
From Lemma \ref{le: preimage}, it holds that for any $s(x)\in \mathcal{S}$, there are exactly $p^{\theta}$ polynomials $\widetilde{f}_m(x) \in \mathcal{F}$ such that
  \begin{displaymath}
         \left\{\begin{aligned}
            &\widetilde{f}_m(x) \equiv (s(x)-\sum_{\ell=1}^{m-1} s_{\ell}(x)) \pmod{m_0(x)},\\
            &\widetilde{f}_m(x) \equiv  c_i^{(m)}(x) \pmod{m_i(x)}~\mathrm{for~all}~i\in \mathcal{B},\\
           \end{aligned} \right.
  \end{displaymath}
where
\[c_i^{(m)}(x)=\begin{cases}
            H_{m}(c_i(x))+u_i^{(m)}(x),~\mathrm{if}~i\in \mathcal{B}\cap [N_{m-1}],\\
            c_i(x),~\mathrm{if}~ i\in \mathcal{B}\cap [N_{m-1}+1,N_m].\\
         \end{cases}\]
Moreover, different $s(x)\in \mathcal{S}$ corresponds to different $\widetilde{f}_m(x) \in \mathcal{F}$. Therefore, the cardinality of $\mathcal{F}$ is $|\mathcal{F}|=p^{d_0}\cdot p^{\theta}=p^{\theta+d_0}$.
\end{proof}

\begin{theorem}[Asymptotic perfectness]\label{Th: our perfectness}
   Our scheme is asymptotically perfect.
\end{theorem}

\begin{proof}
Recall that $\mathcal{B}\subset[n]$, and for each $\ell\in [m]$ we define $\mathcal{B}^{(\ell)} = \mathcal{B} \cap (\bigcup_{w=1}^{\ell} \mathcal{P}_{w})$. Consider the worst-case scenario in which
        \[|\mathcal{B}^{(m)}|< t_{m},~\mathrm{and}~|\mathcal{B}^{(\ell)}|\geq t_{\ell}~\mathrm{for~all}~\ell\in [m-1].\]
For any $s(x)\in \mathcal{S}$, Lemmas \ref{le: preimage} and \ref{le: cardinality of F} give the conditional probability
     \[\mathsf{Pr}(\mathbf{S}=s(x)|\mathbf{V}^\prime_\mathcal{B}=\mathcal{V}^\prime_{\mathcal{B}})=\frac{|\Phi^{-1}(s(x))|}{|\mathcal{F}|}=\frac{p^{\theta}}{p^{\theta+d_0}}=\frac{1}{p^{d_0}}.\]
Therefore,
\begin{displaymath}
             \begin{aligned}
              &\mathsf{H}(\textbf{S}|\textbf{V}^\prime_\mathcal{B})\\
                    =&-\sum_{\mbox{\tiny$\begin{array}{c}
                                       \mathcal{B}\subset[n],|\mathcal{B}^{(m)}|< t_{m},\\
                                      |\mathcal{B}^{(\ell)}|\geq t_{\ell}, \ell\in [m-1]\\
                                       \end{array}$}}
                      \sum_{s(x)\in \mathcal{S}}\mathsf{Pr}(\textbf{V}^\prime_\mathcal{B}=\mathcal{V}^\prime_{\mathcal{B}})
                      \mathsf{Pr}(\mathbf{S}=s(x)|\textbf{V}^\prime_\mathcal{B}=\mathcal{V}^\prime_{\mathcal{B}})\log_2 \mathsf{Pr}(\mathbf{S}=s(x)|\textbf{V}^\prime_\mathcal{B}=\mathcal{V}^\prime_{\mathcal{B}})\\
                    =&\sum_{\mbox{\tiny$\begin{array}{c}
                                       \mathcal{B}\subset[n],|\mathcal{B}^{(m)}|< t_{m},\\
                                      |\mathcal{B}^{(\ell)}|\geq t_{\ell}, \ell\in [m-1]\\
                                       \end{array}$}}
                      \sum_{s(x)\in \mathcal{S}}\mathsf{Pr}(\textbf{V}^\prime_\mathcal{B}=\mathcal{V}^\prime_{\mathcal{B}})\frac{1}{p^{d_0}}\log_2 p^{d_0}\\
                     =&\sum_{\mbox{\tiny$\begin{array}{c}
                                       \mathcal{B}\subset[n],|\mathcal{B}^{(m)}|< t_{m},\\
                                      |\mathcal{B}^{(\ell)}|\geq t_{\ell}, \ell\in [m-1]\\
                                       \end{array}$}}
                      \mathsf{Pr}(\textbf{V}^\prime_\mathcal{B}=\mathcal{V}^\prime_{\mathcal{B}})\log_2 p^{d_0}\\
                    =&\log_2 p^{d_0}.
            \end{aligned}
          \end{displaymath}
 Note that the secret is chosen uniformly at random, so $\mathsf{H}(\textbf{S})=\log_2|\mathcal{S}|=\log_2 p^{d_0}= \mathsf{H}(\mathbf{S}|\mathbf{V}^\prime_\mathcal{B})$. By Lemma \ref{le: loss entropy}, for every $\epsilon_1>0$ there exists a positive integer  $\sigma_1$ such that for all $|\mathcal{S}|=p^{d_0}>\sigma_1$, it holds that
                               \[0<\mathsf{H}(\mathbf{S})-\mathsf{H}(\mathbf{S}|\mathbf{V}_{\mathcal{B}})
                               =\mathsf{H}(\mathbf{S}|\mathbf{V}'_{\mathcal{B}})-\mathsf{H}(\mathbf{S}|\mathbf{V}_{\mathcal{B}})<\epsilon_1.\]
In view of Definition \ref{def: Asymptotically Ideal CHSS}, we conclude that the scheme is asymptotically perfect.
\end{proof}

\begin{theorem}\label{Theorem: summary}
Our scheme is an asymptotically ideal CHSS scheme with an information one when $d_0=d_1=\cdots=d_n$.
\end{theorem}
\begin{proof}
According to Definition \ref{Th: our perfectness}, Theorem \ref{Theorem:Correctness of our scheme} and Theorem \ref{Th: our perfectness}, we get that our scheme is a secure and asymptotically perfect CHSS scheme. Next, we will prove that our scheme has an information one when $d_0=d_1=\cdots=d_n$. Each share of our scheme is given by
         \[c_i(x)=\begin{cases}
            r_i(x),\mathrm{if}~i\in [N_{m-1}],\\
            f_m(x) \pmod {m_i(x)},\mathrm{if}~i\in [N_{m-1}+1,N_m].\\
         \end{cases}\]
Recall that $r_{i}(x)\in \mathbb{F}_{p}[x], i\in [N_{m-1}]$ are random polynomials such that $\deg(r_{i}(x))\leq d_i-1$. The polynomial $f_{m}(x)=s_{m}(x)+\alpha_{\ell}(x)x^{d_0}\in \mathbb{F}_{p}[x]$ is random when $s_{m}(x)\in \mathcal{S}$ and $\alpha_{m}(x)\in \mathcal{G}_{m}$ are chosen randomly. It is easy to check that $\deg(f_{m}(x)\pmod {m_i(x)}<d_i$.
 When $d_0=d_1=\cdots=d_n$, every share space is the same as the secret space $\mathcal{S}=\{g(x)\in \mathbb{F}_p[x]:\deg(g(x))<d_0\}$, which implies that the information rate is $\rho=1$. According to Definition \ref{def: Asymptotically Ideal CHSS}, our scheme is an asymptotically ideal DHSS scheme with an information one when $d_0=d_1=\cdots=d_n$.
\end{proof}

\section{Conclusion}\label{Sec: conclusion}
By using the CRT for polynomial ring and one-way functions, we propose an asymptotically ideal conjunctive hierarchical secret sharing scheme in this work. It has computational security, and permits flexible share sizes. How to reduce the number of public values is our future work.

\section*{Author Contributions}
CRediT: Jian Ding: Conceptualization, Funding acquisition, Writing-original draft; Cheng Wang: Conceptualization, Writing-original draft; Hongju Li: Conceptualization, Validation; Cheng Shu: Validation, Writing-review \& editing; Haifeng Yu: Validation, Writing-review \& editing

\section*{Disclosure Statement }
There are no relevant financial or non-financial competing interests to report.

\end{document}